\documentclass{article}

\usepackage{mathptmx, graphics, graphicx, verbatim, url, amssymb, amsmath, amsfonts, amsthm, latexsym}

\newtheorem{lemma}{Lemma}
\newtheorem{theorem}{Theorem}

\begin{document}

\title{The Massey's method for sport rating: \\ a network science perspective}

\author{Enrico Bozzo \\
Department Mathematics, Computer Science, and Physics \\ 
University of Udine \\
\url{enrico.bozzo@uniud.it} \and 
Massimo Franceschet\\
Department Mathematics, Computer Science, and Physics \\ 
University of Udine \\
\url{massimo.franceschet@uniud.it}}

\maketitle

\begin{abstract}
We revisit the Massey's method for rating and ranking in sports and contextualize it as a general centrality measure in network science.
\end{abstract}

\section{Introduction} \label{sec:introduction}

Rating and ranking in sport have a flourishing tradition. Each sport competition has its own official rating, from which a ranking of players and teams can be compiled.  The challenge of many sports' fans and bettors is to beat the official rating method: to develop an alternative rating algorithm that is better than the official one in the task of predicting future results. As a consequence, many sport rating methods have been developed. Amy N. Langville and Carl D. Meyer even wrote a (compelling) book about (general) rating and ranking methods entitled \textit{Who's \#1?} \cite{LM12}.

In 1997, Kenneth Massey, then an undergraduate, created a method for ranking college football teams. He wrote about this method, which uses the mathematical theory of least squares, as his honors thesis \cite{M97}. Informally, at any given day $k$ of the season, Massey's method rates a team $i$ according to the following two factors: (a) the difference between points for and points against $i$, or point spread of $i$, up to day $k$, and (b) the ratings of the teams that $i$ matched up to day $k$. Hence, highly rated teams have a large point differential and matched strong teams so far. Below in the ranking are teams that did well but had an easy schedule as well as teams that did not so well but had a tough schedule.

In this paper, we review the original Massey's method and explicitly provide the algebraic background of the method. Moreover, we interpret Massey's technique in the context of network science.
The paper is organized as follows. Section \ref{sec:massey} reviews the original Massey's method. In particular, Section \ref{sec:example1} provides a full example of the method, Section \ref{sec:netsci} embeds the method in the context of network science, and Section \ref{sec:algebra} gives the algebraic background of the Massey's method for the curious reader. We review related methods for sport rating in Section \ref{sec:related}. Finally, we discuss alternative (outside the sport context) uses of Massey's methods in Section \ref{sec:discussion}.

\section{The Massey's method for sports ranking} \label{sec:massey}
The main idea of Massey's method, as proposed in \cite{M97}, is enclosed in the following equation:

$$r_i - r_j = y_k$$
where $r_i$ and $r_j$ are the ratings of teams $i$ and $j$ and $y_k$ is the absolute margin of victory for game $k$ between teams $i$ and $j$. If there are $n$ teams who played $m$ games, we have a linear system:

\begin{equation} \label{Massey1}
X r = y
\end{equation}
where $X$ is a $m \times n$ matrix such the k-th row of $X$ contains all 0s with the exception of a 1 in location $i$ and a $-1$ in location $j$, meaning that team $i$ beat team $j$ in match $k$ (if match $k$ ends with a draw, either $i$ or $j$ location can be assigned $1$, and the other $-1$). Observe that, if $e$ denotes the vector of all $1$'s, then $Xe=0$.
Let $M = X^T X$ and $p = X^T y$. Notice that

\begin{equation*}
M_{i,j} = \left\{
    \begin{array}{ll}
     \text{the negation of the \# of matches between } i \text{ and } j & \text{if } i \neq j,\\
      \text{\# of games played by } i  & \text{if } i = j.
    \end{array} \right.
\end{equation*}
and $p_i$ is the signed sum of point spreads of every game played by $i$. Clearly the entries of $p$ sum to $0$, in fact $e^Tp=e^TX^Ty=(Xe)^Ty=0$. The Massey's method is then defined by the following linear system:

\begin{equation} \label{Massey2}
M r = p
\end{equation}
which corresponds to the least squares solution of system (\ref{Massey1}).

We observe how the Massey's team ratings are in fact interdependent. Indeed, Massey's matrix $M$ can be decomposed as

$$M = D - A,$$
where $D$ is a diagonal matrix with $D_{i,i}$ equal to the number of games played by team $i$, and $A$ is a matrix with $A_{i,j}$ equal to the number of matches played by team $i$ against team $j$. Hence, linear system (\ref{Massey2}) is equivalent to

\begin{equation} \label{Massey3}
Dr - Ar = p,
\end{equation}
or, equivalently

\begin{equation} \label{Massey4}
r = D^{-1} (A r + p) = D^{-1} A r + D^{-1} p.
\end{equation}
That is, for any team $i$

$$r_i = \frac{1}{D_{i,i}} \sum_j A_{i,j} r_j + \frac{p_i}{D_{i,i}}.$$
This means that the rating $r_i$ of team $i$ is the sum $r^{(1)}_{i} + r^{(2)}_{i}$ of two meaningful components:

\begin{enumerate}
\item the mean rating of teams that $i$ has matched $$r^{(1)}_{i} = \frac{1}{D_{i,i}} \sum_j A_{i,j} r_j;$$
\item the mean point spread of team $i$ $$r^{(2)}_{i} = \frac{p_i}{D_{i,i}}.$$

\end{enumerate}

\begin{comment}
In the initial part of the season, it is reasonable that the sets of opponents of any two teams have few teams in common, while the overall point spread of any two teams is typically not very far away. Hence, the rating $r$ is mostly influenced by the opponents rating component. As the season goes by, the sets of opponents tend to become uniform and the point spreads tend to diverge, hence the point spread component becomes more and more important relative to the opponents rating one.

Notice also that, if each team has played the same number $k$ of matches, then

$$\begin{array}{ll}
\sum_i r^{(1)}_{i} =  & (1/k) \sum_i \sum_j A_{i,j} r_j  =  \\
                      & (1/k) \sum_j \sum_i A_{i,j} r_j = \\
                      & (1/k) \sum_j r_j \sum_i A_{i,j} = \\
                      & (1/k) \sum_j r_j k = \\
                      & \sum_j r_j = 0
\end{array}$$

and $$\sum_i r^{(2)}_{i} =  (1/k) \sum_i p_i = 0$$ hence both partial components of the ranking $r$ sum to 0.
\end{comment}

How do we solve the Massey's system $M r = p$? Unfortunately, $M = D - A$ is a singular matrix, hence it is not possible to solve this system  computing the inverse of $M$. Notice that the symmetric matrix $A$ can be interpreted as the adjacency matrix of an undirected weighted graph that we denote with $G_A$.  It is well known that $A$ is irreducible if and only if $G_A$ is connected.
Assuming that matrix $A$ is irreducible, or equivalently that the graph $G_A$ is connected, a solution of the system can be obtained as follows: let $\hat{M}$ be the matrix obtained by replacing any row, say the last, of $M$ with a row of all $1$'s, and let $\hat{p}$ be the vector obtained by replacing the last element of $p$ with a 0. Then, solve the system:

\begin{equation} \label{Massey2.1}
\hat{M} r = \hat{p}
\end{equation}
Notice that the added constraint forces the ranking vector $r$ to sum to 0. See Section \ref{sec:algebra} for the mathematical details.

What about the connectivity hypothesis of the graph of matches $G_A$? Without such assumption, the ranking cannot be computed. Indeed, this is not a strong constraint. Assume, realistically, a competition in which there are $n$ teams matching. At each season day each team matches another team not matched before. Hence, for all teams to be matched at any day, $n$ must be even and we have $n/2$ matches each day for a maximum of $n-1$ days before teams match twice. This is known as \textit{round-robin competition}. At season day $k$ the graph $G_A$ is obtained as the disjoint union of $k$ perfect matchings, hence, in particular, it is regular with degree $k$. Moreover, we have the following result:

\begin{lemma} \label{lemma:minmax1}
Given $n > 2$ teams, with $n$ even, competing in a round-robin competition, then:

\begin{enumerate}
\item the minimum number of days for the graph of matches to be connected is 2;
\item the maximum number of days for the graph of matches to be connected is $n/2$.
\end{enumerate}
\end{lemma}

\begin{proof}

We prove item 1. After one season day the graph of matches is not connected being just a perfect matching of the teams. After two day the graph can become a cycle and hence can be connected.

We now prove item 2.
First of all, let us observe that after $n/2-1$ season days the graph can be not connected, but the only possibility is that it is formed by two complete graphs with nodes in $V_1$ and $V_2$ having each $n/2$ nodes. At day $n/2$, all matches are between pairs of teams one of which is in $V_1$ and the other in $V_2$, resulting in a connected graph.
\begin{comment}
Let us assume by contradiction that after $n/2$ days the graph is not connected. Of course at least one of the connected component should have not more than $n/2$ nodes. This is clearly impossible since every node has degree equal to $n/2$ and the nodes of the selected component cannot have a degree exceeding $n/2-1$.
\end{comment}
\end{proof}

Typically, the actual number of days for the match graph to be connected is close to the minimum. For instance, we experimented that in the last 11 editions of the Italian soccer league (Serie A), after 2 or 3 days the match graph is connected, with a mean of 2.6.

Massey also proposed an offensive rating ($o$) and a defensive rating ($d$) characterizing the offensive and defensive strengths of teams. Massey assumes that $r = o + d$, that is, the overall strength of a team is the sum of its offensive and defensive powers. Let's decompose the point spread vector $p = f - a$, where $f$ holds the total number of points scored by each team and $a$ holds the total number of points scored against each team. Then, the equations defining $o$ and $d$ are:

$$
\begin{cases}
Do - Ad =  f \\
Ao - Dd =  a.
\end{cases}
$$
That is for each team $i$:

$$
\begin{cases}
o_i =  \frac{1}{D_{i,i}} (\sum_j A_{i,j} d_j + f_i) \\
d_i =  \frac{1}{D_{i,i}} (\sum_j A_{i,j} o_j - a_i) 
\end{cases}
$$
This means that the offensive rating of team $i$ multiplied by the number of games played by $i$ is equal to the defensive ratings of opponents of $i$ plus the number of points scored by $i$. On the other hand, the defensive rating of team $i$ multiplied by the number of games played by $i$ is equal to the offensive ratings of opponents of $i$ minus the number of points scored against $i$.

Since we know that $r = o + d$, we have that, knowing $r$, the defensive rating $d$ can be computed as:

\begin{equation} \label{MasseyOD}
(D + A) d = D r - f
\end{equation}
and, knowing $d$ and $r$, the offensive rating $o = r - d$.

\begin{comment}
Notice that since $r e = 0$, then $(o+p) e = 0$, hence $o e = - d e$, hence the sum of all offensive strengths is equal to the negation of the sum of all defensive strengths.
\end{comment}

How do we solve system (\ref{MasseyOD})? In fact, it might happen that matrix $D + A$ is singular, hence it cannot be inverted. Assuming that the graph $G_A$ is connected, it holds that $D+A$ is singular precisely when $G_A$ is bipartite.
Let $N = D + A$ and $q = Dr - f$. If $G_A$ is connected and bipartite with node sets $U$ and $V$, we can apply to the system $N d = q$ a perturbation trick similar to the one exploited for the Massey's system $M r = p$. Let $\hat{N}$ be the matrix obtained by replacing any row, say the last, of $N$ with vector $v$ defined as $v_i = 1$ if $i \in U$ and $v_i = -1$ if $i \in V$ and let $\hat{q}$ be the vector obtained by replacing the last element of $q$ with a 0. Then, a solution of system (\ref{MasseyOD}) is obtained solving the following perturbed system:

\begin{equation} \label{MasseyOD.1}
\hat{N} d = \hat{q}
\end{equation}
Again, see Section \ref{sec:algebra} for the mathematical details.
Assuming a round-robin competition, the hypothesis that the match graph $G_A$ is non-bipartite is not a strong one, as proved in the following result.

\begin{lemma} \label{lemma:minmax2}
Given $n > 2$ teams, with $n$ even, competing in a round-robin competition, then:

\begin{enumerate}
\item the minimum number of days for the graph of matches to be non-bipartite is 3;
\item the maximum number of days for the graph of matches to be non-bipartite is $n/2 + 1$.
\end{enumerate}
\end{lemma}

\begin{proof}
We prove item 1. After two days the graph of the matches cannot have cycles of length 3. After $3$ days it is possible to have cycles of length $3$ so that the graph can be non-bipartite.

We now prove item 2. After $n/2$ days the graph of the matches can be bipartite but the only possibility is that it is a complete bipartite graph with node sets $V_1$ and $V_2$ such that $|V_1|=|V_2|=n/2$. At day $n/2 + 1$ all matches are between pairs of nodes either in $V_1$ or in $V_2$. Hence, some cycles of length 3 are formed and therefore the graph is no more bipartite.

\begin{comment}
Let us assume by contradiction that after $n/2+1$ days the graph is bipartite. At least one of the two parts of the graph has not more that $n/2$ nodes but this is impossible since the nodes, and in particular the nodes of the other part, have degree $n/2+1$.
\end{comment}
\end{proof}

Typically, the actual number of days for the match graph to be non-bipartite is close to the minimum. For instance, we experimented that in the last 11 editions of the Italian soccer league (Serie A), within 5 days the match graph is non-bipartite, with a mean of 3.7. Hence, we expect that the Massey's method is fully applicable after few season days in the competition.

\begin{comment}
Given a graph $G$ let $\hat{G}$ the the graph in which there are two nodes $x$ and $\hat{x}$ for every node $x$ of $G$ and there are two edges $(x, \hat{y})$ and $(\hat{x},y)$ for every edge $(x,y)$ in $G$. The following result holds:

\begin{lemma}
Let $G$ be a connected graph. Then $G$ is bipartite if and only if $\hat{G}$ is not connected.
\end{lemma}
\end{comment}

\subsection{Example of Massey's method} \label{sec:example1}

The table below shows the results of 4 matches (numbered 1, 2, 3, 4) involving 4 fictitious teams (labelled A, B, C, D):

\bigskip
\begin{center}
\begin{tabular}{ccccc}
match & team 1 & team 2 & score 1 & score 2 \\ \hline
1 & A & C & 2 & 0 \\
2 & A & D & 3 & 0 \\
3 & B & C & 1 & 1 \\
4 & B & D & 2 & 1 \\
\end{tabular}
\end{center}
The match-team matrix $X$ is given below:

\bigskip
\begin{center}
\begin{tabular}{l|rrrr}
   & A & B & C & D \\ \hline
1 & 1 & 0 & -1 & 0 \\
2 & 1 & 0 & 0 &  -1 \\
3 & 0 & 1 & -1 & 0 \\
4 & 0 & 1 & 0 &  -1 \\
\end{tabular}
\end{center}
The match spread vector $y$ is:

\bigskip
\begin{center}
\begin{tabular}{c|r}
team & y \\ \hline
1 & 2 \\
2 & 3 \\
3 & 0 \\
4 & 1 \\
\end{tabular}
\end{center}
The Massey's matrix $M = X^T X$ is:

\bigskip
\begin{center}
\begin{tabular}{l|rrrr}
   & A & B & C & D \\ \hline
A & 2 & 0 & -1 & -1 \\
B & 0 & 2 & -1 & -1 \\
C & -1 & -1 & 2 & 0 \\
D & -1 & -1 & 0 & 2 \\
\end{tabular}
\end{center}
and the team spread vector $p = X^T y$ is:

\bigskip
\begin{center}
\begin{tabular}{c|r}
team & p \\ \hline
A & 5 \\
B & 1 \\
C & -2 \\
D & -4 \\
\end{tabular}
\end{center}
The resulting Massey's rating is the following:

\bigskip
\begin{center}
\begin{tabular}{c|r}
team & r \\ \hline
A & 1.75 \\
B & -0.25 \\
C & -0.25 \\
D & -1.25 \\
\end{tabular}
\end{center}
Notice that teams B and C have the same rating (and hence ranking), despite the point spread of B (1) is higher than the point spread of C (-2). This is because B against C ended in a draw,
B won the second match, but against the weakest team (D), while C lost the second match, but against the strongest team (A). Indeed, the rating $r$ can be decomposed in the partial ratings $r^{(1)}_{i}$  and $r^{(2)}_{i}$ that follows:

\bigskip
\begin{center}
\begin{tabular}{c|r|r|r}
team  & $r$ & $r^{(1)}_{i}$ & $r^{(2)}_{i}$  \\ \hline
A & 1.75 & -0.75 & 2.5 \\
B & -0.25 &  -0.75 & 0.5 \\
C & -0.25 & 0.75  & -1.0 \\
D & -1.25 & 0.75  & -2.0 \\
\end{tabular}
\end{center}
Hence, $B$ has a better point spread but an easier schedule with respect to $C$. Summing the two rating componentes we get the same rating for teams $B$ and $C$.

Moreover, the rating $r$ can be decomposed in the offensive and defensive ratings that follows:

\bigskip
\begin{center}
\begin{tabular}{c|r|r|r}
team  & $r$ & $o$ & $d$  \\ \hline
A & 1.75  & 1.875 & -0.125 \\
B & -0.25 & 0.875 & -1.125 \\
C & -0.25 & -0.125  & -0.125 \\
D & -1.25 & -0.125  & -1.125 \\
\end{tabular}
\end{center}
where the points for and against vectors are as follows:

\bigskip
\begin{center}
\begin{tabular}{c|r|r|r}
team  & $p$ & $f$ & $a$  \\ \hline
A &  5  & 5 & 0 \\
B &  1 &  3 & 2 \\
C & -2 &  1 & 3 \\
D & -4 &  1 & 5 \\
\end{tabular}
\end{center}
Notice that team $B$ has a better offensive rating than $C$, but a worse defensive rating, and the sum of offensive and defensive powers is the same for both teams. Also, A and C have the same defensive rating, although A has 0 points against and C has 3 points again; again, this because the schedule of C is harder than the schedule of A.

\subsection{Massey's method and network science} \label{sec:netsci}

We observe that Equation (\ref{Massey4}) is close to the equation defining Katz's centrality \cite{K53}:

\begin{equation} \label{Katz}
r = \alpha A r + \beta
\end{equation}
where $\alpha$ is a given constant and $\beta$ is a given (exogenous) vector. 
\begin{comment}
The Katz's equation can be solved by inverting matrix $I - \alpha A$ as soon as $0 < \alpha < 1/\rho(A)$, with $\rho(A)$ the spectral radius of $A$ \cite{N10}.
In particular, making the reasonable assumption that all teams played the same number $k$ of matches, then Equation (\ref{Massey3}) simplifies:

\begin{equation} \label{Massey5}
r = \frac{1}{k} (A r + p)
\end{equation}
which corresponds to Katz's centrality with parameters $\alpha = 1/k$ and $\beta = p/k$. Nevertheless, notice that $\rho(A) = k$ hence $\alpha = 1/\rho(A)$. It follows that Massey's equation is an instance of Katz's centrality, but cannot be computed by inverting matrix $I - \alpha A$.
\end{comment}
The Katz's equation can be solved by inverting matrix $I - \alpha A$ as soon as $\alpha$ does not coincide with the reciprocal of an eigenvalue of $A$ \cite{N10}.
In particular, making the assumption that all teams played the same number $k$ of matches, then Equation (\ref{Massey3}) becomes

\begin{equation} \label{Massey5}
r = \frac{1}{k} (A r + p),
\end{equation}
which corresponds to Katz's centrality with parameters $\alpha = 1/k$ and $\beta = p/k$. Nevertheless, notice that in this setting $k$ is an eigenvalue of $A$. It follows that Massey's equation is an instance of Katz's centrality, but cannot be computed by inverting matrix $I - \alpha A$.

Furthermore, in the following, we propose an electrical interpretation of the Massey's rating. If we view the symmetric matrix $A$ as the adjacency matrix of an undirected graph $G_A$, then $M$ is the Laplacian matrix of the graph $G_A$. The rating vector $r$ defined in system (\ref{Massey2}) is then equivalent to the potential vector over a resistor network defined by $A$ with supply vector $p$ \cite{GBS08}.

Following this metaphor, the resistor network has a resistor between nodes $i$ and $j$ as soon as teams $i$ and $j$ has matched with conductance (reciprocal of resistance) equal to the number $A_{i,j}$ of matches between $i$ and $j$. Sources, that are nodes $i$ with positive supply $p_i$ through which current enters the network, are teams with positive point spread, while targets,  that are nodes $i$ with negative supply $p_i$ through which current leaves the network, are teams with negative point spread. Notice that current entering and leaving the network must be equal, indeed the point spread vector $p$ sums to 0. The potential $r_i$ for node $i$ corresponds to the rating of team $i$: teams with large potential are teams high in the ranking. Moreover, the
current flow through edge $(i,j)$ is, by Ohm's law, the quantity $A_{i,j} (r_i - r_j)$, which corresponds to the rating difference between teams $i$ and $j$ (which is also an estimate of the point spread in a match between $i$ and $j$) multiplied by the number of times they matched: current flows more intensely between teams of different strengths (as measured by Massey's method) that matched many times.

\begin{comment}
The resistor network corresponding to the above example, with nodes labelled with their potentials and edges labelled with their current flows, is given in graph in Figure \ref{fig:resistor}.

\begin{figure}[t]
\begin{center}
\includegraphics[scale=0.60, angle=0]{resistor.pdf}
\end{center}
\caption{The competition resistor network.}
\label{fig:resistor}
\end{figure}

Notice that there is no current flowing between teams B and C, since they share the same rating.
\end{comment}

Finally, it is interesting to analyse what happens to Massey's system at the end of the season, assuming that all $n$ teams matched all other teams once. In this case, the opponents rating component $$r^{(1)}_{i} = -\frac{r_i}{n-1},$$ where we have used the fact that $\sum_i r_i = 0$, and the point spread component $$r^{(2)}_{i} = \frac{p_i}{n-1},$$ hence $$r_i = r^{(1)}_{i} + r^{(2)}_{i} = -\frac{r_i}{n-1} + \frac{p_i}{n-1},$$ and thus $$r_i = \frac{p_i}{n}.$$
The same result is obtained noticing that $$M p = (D-A) p = (n-1) p - A p = np - (p + Ap) = np$$ where we have used the fact that $p$ sums to 0. Hence $p$ is an eigenvector of $M$ with eigenvalue $n$ and hence, once again, $r = p/n$ is a solution of the Massey's system. Hence, the final rating of a team is simply the mean point spread of the team.

It is possible to be a bit more precise about this property of Massey's method by exploiting the properties of the set of eigenvalues, or spectrum, of the Laplacian matrix $M = D - A$. The spectrum
reflects various aspects of the structure of the graph $G_A$ associated with $A$, in particular those related to connectedness.  It is well known that the Laplacian is singular and positive semidefinite (recall that $M=X^TX$ and $Xe=0$)
so that its eigenvalues are nonnegative and can be ordered as follows:
$$\lambda_1=0\le \lambda_2\le \lambda_3 \le \ldots\le \lambda_n.$$
It can be shown that $\lambda_n\le n$, see for example \cite{BH12}.
The multiplicity of $\lambda_1=0$ as an eigenvalue of the Laplacian can be shown to be equal
to the number of the connected components of the graph, see again  \cite{BH12}.
If the graph of the matches is connected or, equivalently,  $M$ is irreducible, as we assume in the following, $\lambda_2\neq 0$ is known as {\em algebraic connectivity} of the graph and is an indicator of the effort to be employed in order to disconnect the graph.

We can write the spectral decomposition of $M$ as $M=UDU^T$ where $U$ is orthogonal and its first column is equal to $e/\sqrt{n}$, and $D={\rm diag}(0$, $\lambda_2$, $\ldots$, $\lambda_n)$. From $Mr=p$ we obtain $r=UD^+U^Tp$ where
$D^+={\rm diag}(0$, $\frac{1}{\lambda_2}$, $\ldots$, $\frac{1}{\lambda_n})$.
Now
$$r-\frac{p}{n}=UD^+U^Tp-\frac{p}{n}=U\Bigr[D^+-\frac{I}{n}\Bigl]U^Tp,$$
where $I$ is the identity matrix.
Observe that the first component of the vector $U^Tp$ is equal to zero so that
$$r-\frac{p}{n}=U\Bigr[D^+-\frac{I}{n}\Bigl]U^Tp=U\Bigr[D^+-\frac{\tilde{I}}{n}\Bigl]U^Tp,$$
where $\tilde{I}={\rm diag}(0,1,\ldots,1)$. If we denote with $\|\cdot\|$ the Euclidean norm we obtain
$$\|r-\frac{p}{n}\|=\|U\bigr[D^+-\frac{\tilde{I}}{n}\bigl]U^Tp\|\le \|p\| \max_{k=2,\ldots,n} \Bigl| \frac{1}{\lambda_k}-\frac{1}{n}\Bigr|\le \|p\|\frac{n-\lambda_2}{n\lambda_2},$$
where we used the fact that the Euclidean norm of an orthogonal matrix is equal to one.
Hence, as the algebraic connectivity $\lambda_2$, as well as the other eigenvalues, approach $n$, that is, as more and more matches are played,
the vector $r$ approaches $p/n$ and the equality is reached when the graph of the matches becomes complete.

\subsection{The underlying linear algebra} \label{sec:algebra}

Let $n\ge 2$. An $n\times n$ matrix $A$ is defined to be strictly diagonally dominant if
$|A_{i,i}|> \sum_{j\neq i} |A_{i,j}|$ for $i=1,\ldots,n$. By using Ger\u{s}gorin theorem it is easy to prove that  if $A$ is strictly diagonally dominant then it is nonsingular, see \cite{HJ13}.
A matrix $A$ is defined to be diagonally dominant if
$|A_{i,i}|\geq \sum_{j\neq i} |A_{i,j}|$ for $i=1,\ldots,n$. A matrix $A$ is defined to be irreducibly diagonally dominant
if it is irreducible, it is diagonally dominant, and for at least one value of $i$  the strict inequality
$|A_{i,i}|> \sum_{j\neq i} |A_{i,j}|$ holds. An irreducibly diagonally dominant matrix is nonsingular, see again \cite{HJ13}.
With a non standard terminology, we define the matrix $A$ to be {\em exactly diagonally dominant} if $|A_{i,i}|=\sum_{j\neq i} |A_{i,j}|$ for $i=1,\ldots,n$. We introduce this concept since the Laplacian and the signless Laplacian of a graph are exactly diagonally dominant. The signless Laplacian of a graph is defined as the matrix whose entries are the absolute values of the entries of the Laplacian \cite{CRS07}. For example, the matrix $M = D-A$ of system (\ref{Massey2})
is the Laplacian of the graph  $G_A$, and the matrix $N = D + A$ of system (\ref{MasseyOD}) is its signless Laplacian.
\begin{comment}
It is useful to extend these definitions to the case where $n=1$. If the only entry of a $1\times 1$ matrix $A$ is different from $0$ the the matrix is strictly diagonally dominant and irreducibly diagonally dominant. Every $1\times 1$ matrix $A$ is diagonally dominant as well as exactly diagonally dominant.

Let $A$ be a symmetric matrix. The matrix $A$ can be interpreted as the adjacency matrix of an undirected weighted graph with possible loops that we denote with $G_A$.  It is well known that $A$ is irreducible if and only if $G_A$ is connected.
\end{comment}
The following lemma points out an interesting property of a symmetric irreducible and exactly diagonally dominant matrix. 

\begin{lemma}\label{newsubmatrices} Let $n\ge2$, $A$  be a symmetric irreducible exactly diagonally dominant
$n\times n$ matrix and $k\in\{1,2,\ldots,n\}$.
The $(n-1)\times(n-1)$  matrix $A(k)$ obtained by deleting from $A$ the $k$-th row and the $k$-th column is nonsingular.
\end{lemma}
\begin{proof}
First of all notice that  $A$ cannot have diagonal elements equal to zero, since, being exactly diagonally dominant, this would imply that $A$ has a row completely equal to zero, and this conflicts with the assumption that $A$ is irreducible.

The matrix $A(k)$ is the adjacency matrix of the subgraph $G_{A(k)}$ obtained by the elimination of node $k$ and of the edges incident in that node from $G_A$. The matrix $A(k)$ can be transformed by a simultaneous permutation of rows and columns into a block diagonal matrix where each diagonal block is relative to one of the connected components of the subgraph. We want to show that each of these blocks is
nonsingular.

If the block has dimension $1$, then it contains one of the diagonal elements of $A$, hence its only entry is different from zero. If the dimension of the block is bigger than $1$, then the block is irreducible because its graph is connected by construction. Moreover, since $A$ is exactly diagonally dominant, each of the blocks of  $A(k)$ is diagonally dominant and the
strict inequality holds for one of the values of the index in each block because at least one of the deleted edges was incident in one of the nodes of the connected component. Hence the block is irreducibly diagonally dominant so that it is nonsingular.

Finally, the matrix $A(k)$ is nonsingular since the diagonal blocks are nonsingular.
\end{proof}

From Lemma \ref{newsubmatrices} we deduce that a symmetric irreducible and exactly diagonally dominant $n\times n$ matrix $A$ has rank at least $n-1$ and more specifically, if we delete
an arbitrary  row from $A$, the remaining rows are linearly independent. Hence, either $A$ is nonsingular, or its nullspace has dimension one. For example the Laplacian of a connected graph is always singular,
while its signless Laplacian is singular if and only if the graph is bipartite \cite{CRS07}.
If $A$ is singular and $v\neq 0$ is such that $Av=0$ then $v$ cannot have entries equal to zero since all the subsets of $n-1$ rows, or equivalently columns, of $A$ are linearly independent. For example for the Laplacian of a connected graph $v=e$ while for its signless Laplacian, in the case where the graph is bipartite,
$v$ has components equal to $1$ or $-1$ in correspondence with the two parts of the graph.

\begin{theorem}\label{theorem:trick}
Let $A$ be a symmetric irreducible and exactly diagonally dominant matrix and let $v\neq 0$ be such that $Av=0$. The following assertions hold true.
\begin{enumerate}
\item The linear system $Ax=f$ is solvable if and only if $f$ is a vector such that $v^Tf=0$.
\item By substituting the $k$-th row of $A$ with $v^T$ we obtain a nonsingular matrix.
\item Let $\hat{A}$ the matrix obtained from $A$ by substituting the $k$-th row of $A$ with $v^T$
and let $\hat{f}$ the vector obtained from $f$ by substituting the $k$-th entry of $f$ with zero.
The solution $x^*$ of the nonsigular system $\hat{A}x=\hat{f}$ is such  that $Ax^*=f$.
\end{enumerate}
\end{theorem}

\begin{proof}

Item 1. Since there exists a vector $v\neq 0$ such that $Av=0$, then $A$ is singular and, as a consequence of Lemma \ref{newsubmatrices}, it has rank $n-1$. Hence, the nullspace of $A$ has dimension $1$ and is generated by the vector $v$. Since $A$ is symmetric, the vector $f$ belongs to the range of $A$ if and only if it is orthogonal to the nullspace of $A$, hence if and only if $v^Tf=0$.

Item 2. The vector $v$ is orthogonal to the rows of $A$, so that it is linearly independent from them. Hence if we substitute one of the rows of $A$ with $v^T$ we obtain a nonsingular matrix.

Item  3.
Observe that $x^*$ satisfies by construction all equations of system $Ax=f$ with the exception of the $k$-th one. Hence, if $A_{j,:}$ denotes the $j$-th row of $A$ we have $A_{j,:}x^*=f_j$ for $j=1,\ldots,k-1,k+1,\ldots,n$. But from $v^TA=0^T$ and $v^Tf=0$ and from the fact that all the entries of $v$ are different from zero we obtain
    $$A_{k,:}=-\sum_{j\neq k}\frac{v_j}{v_k}A_{j,:},\qquad \hbox{and}\qquad f_k=-\sum_{j\neq k}\frac{v_j}{v_k}f_j,$$
    so that
    $$A_{k,:}x^*=-\sum_{j\neq k}\frac{v_j}{v_k}A_{j,:}x^*=-\sum_{j\neq k}\frac{v_j}{v_k}f_j=f_k.$$
\end{proof}

\begin{comment}
Notice that, since by construction $v^T x^*=0$, then
$x^*=G^+f$ where $G^+$ denotes the Moore-Penrose pseudoinverse of $G$, see \cite{BG03}.
\end{comment}

Theorem \ref{theorem:trick} implies that, if $A$ is irreducible, then system (\ref{Massey2}) is solvable; moreover,  Theorem \ref{theorem:trick} justifies the use of system (\ref{Massey2.1}) in order to find a solution.

Let us consider now the system (\ref{MasseyOD}). For a connected graph the signless Laplacian
$N=D+A$ is singular if and and only if the graph is bipartite \cite{CRS07}. Hence for a bipartite graph we have to prove that system  (\ref{MasseyOD}) is solvable. In order to exploit Theorem \ref{theorem:trick} we have only to show that $Dr-f$ is orthogonal to the vector that generates the nullspace of the signless Laplacian.

\begin{lemma} \label{lemma:zero}
Let $U$ and $V$ be the two set of nodes of a bipartite graph $G_A$ with adjacency matrix $A$. Let $v$ be a vector such that $v_i = 1$ if $i \in U$ and $v_i = -1$ if $i \in V$. It holds that $v^T(D r - f) = 0$.
\end{lemma}

\begin{proof}
Let $e_U$ be the vector whose entry $i$ is $1$ if $i \in U$ and 0 if $i \in V$, and $e_V$ be the vector whose entry $i$ is $1$ if $i \in V$ and 0 if $i \in U$. It is sufficient to show that $$e_U^T (Dr -f) = e_V ^T(Dr - f).$$ Recall that $M = (D - A) r = p$, hence $Dr - p = Ar$, and $p = f-a$. Moreover, notice that $e_U^T f = e_V^T a$, since teams in $U$ matched only with teams in $V$.
We have that:

$$
\begin{array}{lcl}
e_V^T (Dr-f) & = & e_V^T D r + e_V^T a - e_V^T f \\
           & = & e_V^T D r - e_V^T p \\
           & = & e_V^T (D r - p) \\
           & = & e_V^T Ar
\end{array}
$$

Finally, notice that $e_U ^TD = e_V^T A$.
\end{proof}

Theorem \ref{theorem:trick} and Lemma \ref{lemma:zero} imply that, if $A$ is irreducible, then system (\ref{MasseyOD}) is solvable; moreover, Theorem \ref{theorem:trick} justifies the use of system (\ref{MasseyOD.1}) in order to find a solution in the case where the graph of $A$ is bipartite.

\section{Related literature} \label{sec:related}

In this section we review some popular alternatives of the original Massey method. 
Keener's method \cite{Ke93} is a recursive technique based on Perron-Frobenius theorem. Let $S_{i,j}$ be the number of points, or any other relevant statistics, that $i$ scores against $j$ and $A_{i,j}$ be the strength of $i$ compared to $j$. The strength $A_{i,j}$ can be interpreted as the probability that $i$ will defeat team $j$ in the future.  For instance, $A_{i,j} = S_{i,j} / (S_{i,j} + S_{j,i})$ or, using Laplace's rule of succession, 
$A_{i,j} = (S_{i,j} + 1) / (S_{i,j} + S_{j,i} + 2)$. Notice that in both cases $0 \leq A_{i,j} \leq 1$ and $A_{i,j} + A_{j,i} = 1$. The Keener's rating $r$ is the solution of the following recursive equation: $$A r = \lambda r,$$ where $\lambda$ is a constant and $A$ is the team strength matrix, or
$$r_i = \frac{1}{\lambda} \sum_j A_{i,j} r_j.$$ Hence, the thesis of Keener is that team $i$ is strong if it defeated strong teams. 

This intriguing recursive definition has been discovered and rediscovered many times in different contexts. It has been investigated, in chronological order, in econometrics, sociometry, bibliometrics, Web information retrieval, and network science \cite{F11-CACM}. In particular, it is the basis of the PageRank algorithm used by Google search engine to rank Web pages \cite{BP98}. Assuming that matrix $A$ is nonnegative and irreducible (or equivalently that the graph of $A$ is connected), Perron-Frobenius theorem guarantees that a rating solution exists. 

The offense-defense method is described in Langville and Meyer's book \cite{LM12}, where it is applied to the problem of rating and ranking of sport teams. A slightly modified method is proposed and investigated in a previous work by Knight \cite{Kn08}. The method lies on Sinkhorn-Knopp theorem. Let $A$ be a strength matrix as defined for Keener's method, that is, $A_{i,j}$ is the probability that $i$ will defeat team $j$ in the future. For any team, two mutually dependent scores
known as offensive $o$ and defensive $d$ ratings are defined as follows:

$$\begin{array}{lcl}
o & = & A d^{\div} \\
d & = & A^T o^{\div} .
\end{array}$$

where $d^{\div}$ is a vector whose components are the reciprocals of those of $d$ (and similarly for $o^{\div}$). Hence:

$$\begin{array}{lcl}
o_i & = &  \sum_k A_{i,k} / d_k \\
d_i & = &  \sum_k A_{k,i} / o_k.
\end{array}$$

It means that a team has high offensive strength (that is, high offensive rating) if it defeated teams with a high defensive strength (that is, low defensive rating), and a team has high defensive strength (that is, low defensive rating) if it was defeated by teams with a high offensive strength (that is, high offensive rating). Assuming that matrix $A$ has total support (or equivalently that the graph of $A$ is regularizable), Sinkhorn-Knopp theorem guarantees that a rating solution exists. The offense-defense method can be regarded as a non-linear (reciprocal) version of Kleinberg's method HITS \cite{K99}.

Finally, the Elo's system is an iterative method coined by the physics professor and excellent chess player Arpad Elo. Let $S_{i,j}$ be the score of team $i$ against team $j$; for instance, in chess a win is given a score of 1 and a draw a score of 1/2 (and a defeat a score of 0). Let $\mu_{i,j}$ be the number of points that team $i$ is expected to score against team $j$; this is typically computed as a logistic function of the difference of ratings between the players, for instance, $$\mu_{i,j} = \frac{1}{1 + 10^{-d_{i,j} / \zeta}},$$ where $d_{i,j} = r_i(old) - r_j(old)$ and $\zeta$ is a constant (in the chess world $\zeta = 400$). Then, when teams $i$ and $j$ match, the new rank $r_i(new)$ of team $i$ is updated as follows (and similarly for $j$): $$r_i(new) = r_i(old) + \kappa (S_{i,j} - \mu_{i,j}),$$ where $\kappa$ is a constant (for instance, in chess $\kappa = 25$ for new players). 
Hence, beating a stronger player has a larger reward than beating a weaker one. According to the movie \textit{The social network} by David Fincher, it appears that the Elo's method formed the basis for rating people on Zuckerberg's Web site Facemash, which was the predecessor of Facebook. 

\section{Discussion} \label{sec:discussion}

We believe that Massey's method is valid as a centrality measure even outside the sport world. Immagine a weighted directed network in which nodes have an advantage in having weighty edges exiting from the node, while they have a disadvantage in having weighty edges entering to the node. For instance, this holds when nodes are countries and edges corresponds to valued transfer of goods between countries: an edge weighted $k$ from A to B means that country A exports to country B for a value of $k$ (or B imports from A for a value o $k$). As another example, imagine if nodes are financial actors (like banks) and edges are money that have been loaned between banks: an edge weighted $k$ from A to B means that bank A loaned to bank B a value of $k$ (or B borrowed from A a value o $k$).

We can view such a weighted directed network as a competition between nodes of the network as follows:

\begin{itemize}
\item for any pair of nodes $i$ and $j$ for which there is an edge $(i,j)$ with weight $k_1$ as well as the reciprocal edge $(j,i)$ with weight $k_2$, we have a match between $i$ and $j$ in which $i$ scores $k_1$ points and $j$ scores $k_2$ points;

\item for any pair of nodes $i$ and $j$ for which there is an edge $(i,j)$ with weight $k_1$ but there is no reciprocal edge $(j,i)$, we have a match between $i$ and $j$ in which $i$ scores $k_1$ points and $j$ scores $0$ points;

\item for any pair of nodes $i$ and $j$ that are not connected by an edge (in any direction), we have no match between $i$ and $j$.
\end{itemize}

Having defined the matches of the competition, we can apply Massey's method (or any alternative sport rating method) to estimate the importance of nodes in the network. For instance, in the case of Massey's method and the banking scenario illustrated above, we are looking for ratings of banks such that $r_i - r_j = y_k$, where $y_k$ is the financial balance (credit minus debit) between banks $i$ and $j$. Hence if $j$ is highly indebted with $i$, then the difference in rankings between $i$ and $j$ is large (in favor of $i$).

\bibliographystyle{plain}
%\bibliography{bibliography}

\end{document}